\documentclass[a4paper,11pt]{article}

\usepackage{amsfonts,amsmath,amssymb,amsthm,epsfig,psfrag,verbatim}
\usepackage[all]{xy}
\usepackage[margin=1in,nomarginpar,nohead,nofoot,a4paper,footskip=1cm]{geometry}

\newtheorem{theorem}{Theorem}

\newcommand{\EE}{{\mathbf{E}}}
\newcommand{\PP}{{\mathbf{P}}}
\newcommand{\RR}{{\mathbb{R}}}

\newtheorem{definition}{Definition}

\begin{document}
\title{Community Detection in the \\
Labelled Stochastic Block Model}
\author{Simon Heimlicher\thanks{part of this author's work has been conducted while employed at Technicolor.}\\ University of Massachusetts\\ Amherst, USA\\  \and Marc Lelarge\\ INRIA-ENS\\ Paris,
  France \and Laurent Massouli\'e${}^{\small *}$\\INRIA-Microsoft Research Joint Centre\\ Palaiseau, France }
\date{}
\maketitle
\begin{abstract}
We consider the problem of community detection from observed
interactions between individuals, in the context where multiple types
of interaction are possible. We use {\em labelled} stochastic block
models to represent the observed data, where labels correspond to
interaction types. Focusing on a two-community scenario, we conjecture
a threshold for the problem of reconstructing the hidden
communities in a way that is correlated with the true partition.
To substantiate the conjecture, we prove that the given threshold correctly identifies a transition on the behaviour of belief propagation from insensitive to sensitive. We further prove that the same threshold corresponds to the transition in a related inference problem on a tree model from infeasible to feasible. Finally, numerical results using belief propagation for community detection give further support to the conjecture.
\end{abstract}

\section{Introduction}%
\label{sec:intro}

Community detection consists in the identification of underlying clusters of individuals with similar properties in an overall population. It is relevant in vastly diverse contexts such as biology and sociology, where one might want to classify proteins or humans respectively, based on their interactions. Most methods assume interactions to be described by a graph, whose edges represent pairs of individuals known to interact. They then amount to graph clustering, with potentially distinct flavours: assortative communities see more interactions within than across communities, while the opposite holds in the disassortative case. 

The stochastic block model provides a versatile model of community structure, allowing representation of diverse scenarios and analytical comparison of candidate algorithmic detection procedures. In this model, nodes are partitioned into blocks,  and an edge is present between any two nodes with a probability depending only on the blocks to which each of the two nodes belong. Despite its simplicity, this model already displays rich behaviours, some of which are not yet fully understood. One phenomenon of practical interest consists in a phase transition from a situation where the graph of interactions does not reveal any structure, to one where it reflects some of the underlying structure. In the latter case, algorithmic procedures such as belief propagation can perform non-trivial classifications of  nodes. 

The simplest example of this situation consists in a model with $n$ nodes partitioned into two equal-size blocks, and where two nodes are connected with probability $a/n$ or $b/n$ depending on whether they belong to the same block or not. Then it is known that the Condition
\begin{equation}\label{cond1} 
(a-b)^2> 2(a+b)
\end{equation} 
is necessary for reconstruction, i.e. cluster in a way correlated with
the true partition. Mossel et al.~\cite{Mossel:2012ve} have indeed shown that, if it is violated, then the distribution of the observed graph is absolutely continuous with respect to that of an unstructured fully symmetric random graph without underlying block structure. When this condition holds, it is conjectured by Decelle et al.~\cite{Decelle:2011ve} that the underlying block structure can at least partially be recovered by belief propagation. Beyond their theoretical interest, such threshold phenomena also have some practical implications: they indicate what amount of downsampling or perturbation of original data can be tolerated before all useful information is lost. 

Three elements support the conjecture that under Condition (\ref{cond1}) community detection is possible. First, Decelle et al.~\cite{Decelle:2011ve} show that it implies {\em sensitivity of belief propagation to noise}. Second, it is known to correspond to a certain reconstruction threshold for a model of infinite random trees, whose structure locally resembles that of the stochastic block model. Third, numerical evaluations indicate the ability of belief propagation to retrieve some of the underlying structure under (\ref{cond1}).

In the present work, we initiate an investigation of similar phenomena in the more general context of {\em labelled} stochastic block models. In such models the observation of an interaction between any two individuals is enriched with a label which represents that interaction's particular type. Many applications of community detection naturally feature such labels. Protein-protein chemical reactions may be exothermic or endothermic; (movie-user) associations in collaborative filtering typically come with user ratings; email exchanges may be cold, formal, or familiar; etc.  


Our main contribution consists in a generalization of Condition (\ref{cond1}) describing the transition from unidentifiable to identifiable to the context of labelled stochastic block models. Specifically, after introducing necessary notation and our main conjecture in Section~\ref{sec:model}, we show in Section~\ref{sec:phtrBP} that our generalized condition corresponds to the transition between insensitivity to sensitivity in belief propagation. We then show in Section~\ref{sec:rec_tree} that it also coincides with the reconstruction threshold for the corresponding labelled tree model. The conjecture is further validated numerically in Section~\ref{sec:numerical} where belief propagation is shown to achieve useful detection only above the threshold. Conclusions are drawn in Section~\ref{sec:conclusion}.

\section{Related Work}%
\label{sec:related}
Several works address community detection in the un-labelled stochastic block model. The two main approaches are based on belief propagation and spectral methods. Spectral methods typically ensure consistent reconstruction in regimes with high ($\omega(1)$) average degree. An early reference is McSherry~\cite{McSherry:2001}. More recently Rohe et al.  \cite{Rohe:2011wq} use Laplacian spectra, and address growing numbers of communities, but still require high ($\omega(1)$) connectivity. Decelle et al. \cite{Decelle:2011ve} rely on belief propagation, and heuristically determine a threshold for detectability in a ``sparse'' regime, where node degrees are of order 1. 

The related problem of tree reconstruction has initially been
considered by Evans et al. \cite{ekps00}, who identified a
threshold on the tree's mean degree above which reconstruction is
feasible through a simple ``census'' method. This threshold was later
shown to correctly identify the onset of ``robust reconstruction'' by
Janson and Mossel \cite{janson:2004}.  We refer to \cite{mos04} for a
survey of this area.

A complete understanding of the relation between thresholds for community detectability in block models and reconstrruction in associated tree models is still missing. See, however, Gerschenfeld and Montanari~\cite{Gerschenfeld:2007} for conditions under which the two thresholds coincide. For the symmetric two-community case, Mossel et al.~\cite{Mossel:2012ve} show that the threshold for community detectability is at least as large as that for tree reconstruction; Coja-Oghlan~\cite{Coja-Oghlan:2010} determines an upper bound on the threshold for community detection, that is believed to be loose. 

In contrast, to the best of our knowledge the problem of community detection and tree reconstruction in the {\em labelled} case has not been explicitly considered in the literature.
\section{Model description and main conjecture}%
\label{sec:model}
In the sequel we focus on the simplest non-trivial labelled stochastic block model, which is defined as follows. A total of $n$ nodes are split into two equal-size blocks, namely block 0 and block 1. The type of any given node $i\in\{1,\ldots,n\}$ refers to the block it belongs to, and is denoted by $\sigma_i\in\{0,1\}$. Any two nodes $i$, $j$ are related with probability $a/n$ if $\sigma_i=\sigma_j$, and with probability $b/n$ otherwise, where $a$, $b$ are two positive constants. Furthermore, given any two related nodes $i$, $j$, one observes a label $L_{ij}$ taking its values in some finite set ${\mathcal{L}}$. Label $L_{ij}$ is drawn from distribution $\{\mu(\ell)\}_{\ell\in{\mathcal{L}}}$ if $\sigma_i=\sigma_j$, and from distribution $\{\nu(\ell)\}_{\ell\in{\mathcal{L}}}$ otherwise. 

Note that the present model generalizes the one studied in Mossel et
al.~\cite{Mossel:2012ve}, to which it reduces when the labels do not
bring extra information relative to the types of the underlying
nodes, that is when $\mu(\ell)\equiv\nu(\ell)$. 
In this context, we make the following conjecture:

{\bf Conjecture:} In the labelled stochastic block model with two symmetric blocks, connectivity parameters $a$, $b>0$ and label distributions $\mu$, $\nu$, reconstruction is infeasible if $\tau<1$, while it is feasible when $\tau>1$, where the threshold value $\tau$ is defined as
\begin{equation}\label{eq:tau}
\tau:=\lambda\sum_{\ell\in{\mathcal{L}}} \frac{a\mu(\ell)+b\nu(\ell)}{a+b}\left(\frac{a\mu(\ell)-b\nu(\ell)}{a\mu(\ell)+b\nu(\ell)}\right)^2,
\end{equation} 
and $\lambda:=(a+b)/2$ is the mean degree in the corresponding block model.

Note that this extends the conjecture made for the un-labelled case in \cite{Mossel:2012ve}, as the Condition $\tau>1$ simplifies to (\ref{cond1}) when $\mu(\ell)\equiv\nu(\ell)$. We will now establish several results supporting this conjecture.

\section{Phase transition for belief propagation sensitivity}%
\label{sec:phtrBP}

We first introduce a labelled tree which can be coupled with the
original graph, see Proposition 5.2 in \cite{Mossel:2012ve} (the only
difference here is the addition of labels on edges).
Consider the following random tree version of the reconstruction
problem. Starting from a root node $r$ with type $\sigma_r
\in\{0,1\}$, consider a branching process with the following
characteristics. Each node $i$ with type $\sigma_i$ gives birth to a
number of children of type $t=\sigma_i$ with Poisson distribution
$\hbox{Poi}(a/2)$ and to a number of children of type $t=1-\sigma_i$
with Poisson distribution $\hbox{Poi}(b/2)$. Conditional on the types
$(t,t')$ of a (parent-child) pair $(i,j)$, a label $L_{ij}$ is
attached to the edge $(i,j)$, drawn independently of everything else
with distribution $\mu$ if $t=t'$, and with distribution $\nu$ if $t\ne
t'$.

Consider now such a tree up to depth $d$, that we denote ${\mathcal{T}}_d$. For each node $i\in{\mathcal{T}}_d$, denote by ${\mathcal{T}}_d(i)$ the subtree rooted at node $i$, together with its labels. Let $X_i=\PP(\sigma_i=1|{\mathcal{T}}_d(i))$, and 
$$
R_i:=\frac{X_i}{1-X_{i}}\cdot
$$
Bayes formula entails that
$$
R_i=\prod_{j\hbox{ child of }i}\frac{X_j a \mu(L_{ij})+(1-X_j)b\nu(L_{ij})}{X_j b \nu(L_{ij})+(1-X_j)a\mu(L_{ij})}\cdot
$$
This readily reduces to a recursion in terms of the random variables $R_j$:
$$
R_i=\prod_{j\hbox{ child of }i}\frac{R_j a \mu(L_{ij})+b\nu(L_{ij})}{R_j b \nu(L_{ij})+a\mu(L_{ij})}\cdot
$$
It also follows at once from these expressions that if one starts from uniform beliefs ($X=1/2$ or equivalently $R=1$ on the leaves), then uniform beliefs constitute a fixed point.

Following Decelle et al.~\cite{Decelle:2011ve}, we introduce the following notion of robustness to noise for  this fixed point:
\begin{definition}Assume that belief ratios $R$ for leaf nodes at depth less than $d$ are fixed to 1. 
The belief ratio $R_r$ at root $r$ is then determined by induction from the belief ratios $R_j$ of nodes at depth $d$, i.e. $j\in{\partial\mathcal{T}_d}$, through a map $F_d$: $R_r=F_d(R_j,j\in{\partial\mathcal{T}_d})$.

The infinitesimal sensitivity $\chi(d)$ of the root belief $R_r$ to noise at depth $d$ is defined as 
\begin{equation}
\chi(d)=\lim_{\epsilon\to 0} \frac{1}{\epsilon^2}\hbox{Var}\left(F_d(1+\epsilon \xi_j, j\in{\partial\mathcal{T}_d})|{\mathcal{T}}_d\right),
\end{equation}
where the $\xi_j$ are i.i.d. unit variance random variables.
The fixed point $R\equiv 1$ is then said to be insensitive to noise if $\lim_{d\to\infty}\chi(d)=0$, and sensitive to noise if $\lim_{d\to\infty}\chi(d)=+\infty$.
\end{definition}

With these definitions at hand, we are ready to state the following
\begin{theorem}\label{th:BPsens}
Let $\tau$ be defined by expression (\ref{eq:tau}). Then the fixed point $R\equiv 1$ is insensitive to noise if $\tau<1$ and sensitive to noise if $\tau>1$.
\end{theorem}
Before we prove the theorem, let us comment on the implications. 
As conjectured in Decelle et al. in the case of un-labelled data,
community detection is infeasible in an instance which is insensitive
to noise, while it is feasible (i.e. some reconstruction classifying
correctly more than half the nodes) in an instance that
is sensitive to noise. This leads us to state the conjecture in
Section \ref{sec:model}.

Before proving Theorem \ref{th:BPsens} we need a technical result.
Consider thus a branching process with Poisson offspring
distribution with mean $\lambda$ for some $\lambda>1$. In addition,
each parent-child edge in the corresponding branching tree is endowed
with a real weight. 
All weights $W$ are sampled in an i.i.d. fashion with moment
generating function: $\varphi(\theta) =\EE\left[e^{\theta
    W}\right]<\infty$.

We let $N(d)$ denote the number of descendants in the $d-$th
generation. We further let $N^+(d,s)$ (resp. $N^-(d,s)$) denote the number of such
descendants whose sum of weights along the path from the ancestor to
them is larger (resp. smaller) than $ds$. 

Let us now introduce the so-called rate function $h$ as follows. First, we let 
$$
h_0(x):=\sup_{y\in \RR}\left(xy-\log(\varphi(y))\right).
$$
This is the so-called Cram\'er transform of the weights distribution,
which by Cram\'er's theorem determines the behaviour of large
deviations of empirical means $(1/d)\sum_{t=1}^d W_t$  of i.i.d. weights from their expectation $\bar{w}:=\varphi'(0)$. Let now $w^-$ and $w^+$ be defined as 
$$
\left\{
\begin{array}{ll}
w^+&=\inf\{x\ge \bar{w}: h_0(x)\ge \log\lambda\},\\
w^-&=\sup\{x\le \bar{w}: h_0(x)\ge \log\lambda\}.
\end{array}
\right.
$$
We then let
\begin{equation}
h(x):=\left\{
\begin{array}{ll}
h_0(x)&\hbox{if }x\in[w^-,w^+],\\
+\infty&\hbox{otherwise.}
\end{array}
\right.
\end{equation}
We are now ready to state the following
\begin{theorem}\label{theorem1}
For any $x\ge \bar{w}$, $x\ne w^+$, on the event that the branching process survives indefinitely, one has the almost sure convergence
\begin{equation}\label{up}
\lim_{d\to\infty}\left(N^+(d,x)\right)^{1/d}=\lambda e^{-h(x)}.
\end{equation}
Similarly, for all $x\le \bar{w}$, $x\ne w^-$, on the event that the branching process survives indefinitely, one has
\begin{equation}\label{down}
\lim_{d\to\infty}\left(N^-(d,x)\right)^{1/d}=\lambda e^{-h(x)}.
\end{equation} 
\end{theorem}
\begin{proof}
We only prove (\ref{up}), as the other property (\ref{down}) is shown similarly. Consider first the case where $x>w^+$. The expectation of the summation in (\ref{up}) reads
$$
\EE N^+(d,x)= \lambda^d \PP\left(\sum_1^d W_t\ge x d\right).
$$
Chernoff's bound implies that this is no larger than
$e^{d(\log\lambda-h_0(x))}$. Being an integer-valued random variable,
the summation is then positive only with probability at most
$e^{d(\log\lambda-h_0(x))}$. By Borel-Cantelli's lemma, it is then
positive only for finitely many $d$'s. Thus the limit in (\ref{up}) is
0, as announced.

The case where $x\in[\bar{w},w^+)$ follows from a general result for
branching random walks \cite{big92}. Indeed consider the random
measure on $\mathbb{R}$:
\begin{eqnarray*}
Z^{(d)} = \sum_{i=1}^{N(d)} \delta(X_i),
\end{eqnarray*}
where $X_i$ is the sum of the weigths along the path from the ancestor
to the $i$-th individual in generation $d$. Note that we have $N^+(d,x)=Z^{(d)}[xd;\infty)$.

It is well-known that
\begin{eqnarray*}
M^{(d)}(x) := \left( \lambda \varphi(x)\right)^{-d}\int
e^{xy}Z^{(d)}(dy),
\end{eqnarray*}
is a positive martingale and hence has an almost sure limit
$M(x)$ as $d$ tends to infinity.
For $x\in (w^-,w^+)$, as shown in \cite{big92}, the limit $M(x)$ is stricly positive if
the process survives.
Then Theorem 4 in \cite{big92} implies that for any fixed $0<h$ as $d$
tends to infinity:
\begin{eqnarray*}
\left(Z^{(d)}[xd-h,xd+h]\right)^{1/d} \to \lambda e^{-h_o(x)}. 
\end{eqnarray*}
This clearly gives a lower bound to (\ref{up}). The upper bound is
easily obtained by the following argument:
\begin{eqnarray*}
N^+(d,x) = Z^{(d)}[xd,\infty) &\leq& \int e^{\theta(y-xd)}Z^{(d)}(dy)\\
&=& e^{-\theta xd}M^{(d)}(\theta) \lambda^d\varphi(\theta)^d,
\end{eqnarray*}
minimizing over $\theta<w^+$ (which ensures that $\lim_{d\to
  \infty}M^{(d)}(\theta)=M(\theta)>0$) gives the desired result.
\end{proof}

Let us now prove Theorem \ref{th:BPsens}. We first determine an expression for the infinitesimal sensitivity $\chi(d)$. Using linearization, we have that 
$$
\chi(d)=\sum_{j\in{\mathcal{F}}(d)}\prod_{(uv)\in\hbox{path }(j\sim r)}\left(\frac{\partial}{\partial R}\left. \frac{R a\mu(L_{uv})+b \nu(L_{uv})}{a\mu(L_{uv})+R b \nu(L_{uv})}\right|_{R=1}\right)^2.
$$
The derivative in the above formula reads
$$
\frac{\partial}{\partial R}\left. \frac{R a\mu(L_{uv})+b \nu(L_{uv})}{a\mu(L_{uv})+R b \nu(L_{uv})}\right|_{R=1}
=\frac{a\mu(L_{uv})-b\nu(L_{uv})}{a\mu(L_{uv})+b\nu(L_{uv})}\cdot.
$$
Let us denote the absolute value of this expression by $e^{W_{uv}}$ for some suitably defined weight $W_{uv}$, so that
$$
\chi(d)=\sum_{j\in{\mathcal{F}}(d)}\exp\left(\sum_{(uv)\in\hbox{path }(r\sim j)}2 W_{uv}    \right).
$$
Note that in the present model, thanks to symmetry between the two classes $0,1$, the labels $L_{uv}$ are i.i.d., with probability distribution $\PP(L=\ell)=\frac{a\mu(\ell)+b\nu(\ell)}{a+b}$.

We are thus in the setup of Theorem~\ref{theorem1}, with a dsitribution for the weights suitably derived from this label distribution and the transform $W=\log(|a\mu(L)-b\nu(L)|/(a\mu(L)+b\nu(L))$.

We then have, from Theorem~\ref{theorem1}, applying the Laplace method, the exponential equivalent:
\begin{equation}\label{logequ}
\frac{1}{d}\log\chi(d)\sim \log\lambda+\sup_{x\in\RR}(2x-h(x)).
\end{equation} 
Consider the modified expression $\sup_x(2x-h_0(x))$, and let $x^*$ denote the point attaining this supremum. By convexity of $h_0$ and the fact that it achieves its minimum at $\bar{w}$, necessarily $x^*\ge \bar{w}$. This supremum equals $\log\EE e^{2W}$ by convex duality. Note also that $x^*\le 0$, since the support of the distribution of $W$ is in $\RR^-$. Consider first the case where $\tau>1$, or equivalently,
$$
\log\lambda + \log\EE e^{2W}>0.
$$
We then have $h_0(x^*)=2x^*-\log\EE e^{2W}< \log\lambda$ by the above condition, so that $h_0(x^*)=h(x^*)$. Thus the logarithmic equivalent~(\ref{logequ}) reads $\log(\tau)$ and is strictly positive. We thus have sensitivity to perturbations.

Consider next the case where $\tau<1$, i.e. $\log\EE e^{2W}<-\log\lambda$. In that case, the logarithmic equivalent (\ref{logequ}) is upper-bounded by $\log(\tau)$ and is thus strictly negative. Insensitivity to perturbations follows.

\section{Phase transition for reconstructability on labelled trees}%
\label{sec:rec_tree}

In this section, $\mathcal{T}$ is an infinite tree with types
$\sigma\in \{0,1\}$ on its vertices and labels $L$ on its edges.
To have consistent notation with previous section, a child has the
same type as its parent with probability $\frac{a}{a+b}$.
Given that the child has the same type as its parent, its label is
distributed as $\mu(\ell)$, otherwise it is distributed according to
$\nu(\ell)$.
Note that if $\mathcal{T}$ is a realization of a Galton-Watson tree with offspring
distribution $\hbox{Poi}\left(\frac{a+b}{2}\right)$ conditioned on
non-extinction, we get exactly the same tree model as in the previous
section. 
In this section, the underlying tree is fixed (i.e. non-random) so
that the only randomness considered here is associated with the types
of the vertices and the labels of the edges.

We denote by $\PP_{0}$ and $\EE_{0}$ the
probability distribution and expectation conditional on the labels of
the edges of the tree.
We define the function $\epsilon:\mathcal{L}\to[0,1/2]$ by
\begin{eqnarray*}
\epsilon(\ell) = \frac{b\nu(\ell)}{a\mu(\ell)+b\nu(\ell)}.
\end{eqnarray*}a
If $j$ is a child of $i$, we have
\begin{eqnarray*}
\PP_{0}\left( \sigma_i\neq\sigma_j\right)= \epsilon(L_{ij}).
\end{eqnarray*}
We now give an alternative description of the random types of the
vertices of the tree when the labels of the edges are known,
i.e. conditionally on the labels. At the root $r$ of the tree ${\mathcal{T}}$ a
binary random variable is chosen uniformly at random. This type is then
propagated, with error, throughout the tree as follows: the child $j$
of the vertex $i$ receives the type of $i$ with probability
$1-\epsilon(L_{ij})$, and the opposite type with probability
$\epsilon(L_{ij})$. These events at the vertices are statistically
independent. This model has been studied in information theory,
mathematical genetics and statistical physics when the function
$\epsilon$ is constant. We refer to \cite{ekps00} for references.

Suppose we are given the types that arrived at the $d$-th level
$\partial\mathcal{T}_d$ of the tree $\mathcal{T}$. Observing the labels of the
edges and using optimal reconstruction strategy (maximum likelihood),
the probability of correctly reconstructing the original type at the
root is denoted by
$\left(1+\Delta(\mathcal{T},d)\right)/2$, where clearly
$\Delta(\mathcal{T},d)\geq 0$.

For an infinite tree $\mathcal{T}$, we denote by $\lambda=\lim\sup_d
|\partial\mathcal{T}_d|^{1/d}$ its growth rate. Note that our notation is
consistent with the previous section, as in the case where $\mathcal{T}$ is
a realization of a Galton-Watson tree with offspring
distribution $\hbox{Poi}\left(\frac{a+b}{2}\right)$, 
$\lambda=\frac{a+b}{2}$ a.s. We still define $\tau$ by the expression (\ref{eq:tau}).
Adapting the argument of \cite{ekps00}, we are able to show:
\begin{theorem}
Let $\mathcal{T}$ be an infinite labelled tree with root $r$ as
defined above.
Consider the problem of reconstructing the type of the root $\sigma_r$ from the types at the
$d$-th level $\partial\mathcal{T}_d$ of $\mathcal{T}$ and the labels on the
tree.
\begin{enumerate}
\item If $\tau>1$ then $\inf_{d\geq 1}\Delta(\mathcal{T},d)>0$;
\item 
If $\tau<1$ then $\inf_{d\geq 1}\Delta(\mathcal{T},d)=0$.
\end{enumerate}
\end{theorem}
\begin{proof}
Following \cite{ekps00}, we derive a lower bound for
$\Delta(\mathcal{T},d)$ in terms of the effective
electrical conductance from the root $r$ to $\partial\mathcal{T}_d$ and an
upper bound which is the maximum flow from $r$ to $\partial\mathcal{T}_d$ for
certain edge capacities. We refer to \cite{lyons-book} for background
on these notions.

For the conductance lower bound, we follow Section 5 of \cite{ekps00}
and for each edge $(i,j)$, $j$ a child of $i$,  we define $\theta_{ij}=1-2\epsilon(L_{ij})=\frac{a\mu(L_{ij})-b\nu(L_{ij})}{a\mu(L_{ij})+b\nu(L_{ij})}$
and then assign the resistance
\begin{eqnarray*}
R_{ij} = (1-\theta_{ij}^2)\prod_{(uv)\in\hbox{path }(r\sim j)}\theta_{uv}^{-2},
\end{eqnarray*}
where $\hbox{path }(r\sim j)$ is the path from the root $r$ to
node $j$. We also define for each vertex $i$
\begin{eqnarray*}
\Theta_i = \prod_{(uv)\in\hbox{path }(r\sim i)}\theta_{uv}.
\end{eqnarray*}
By Theorem 1.2' and 1.3' of \cite{ekps00}, we have
\begin{eqnarray*}
\Delta(\mathcal{T},d)\geq
\frac{1}{1+\mathcal{R}_{\hbox{eff}}(r\leftrightarrow
  \partial\mathcal{T}_d)}&\mbox{ and }&
\Delta(\mathcal{T},d)^2\leq 2\sum_{i\in \partial\mathcal{T}_d}\Theta_i^2,
\end{eqnarray*}
where $\mathcal{R}_{\hbox{eff}}(r\leftrightarrow \partial\mathcal{T}_d)$ is
the effective resistance between the root $r$ and the $d$-th level of
the tree.

We first prove our second claim. Note that
\begin{eqnarray*}
\EE_0\left[\theta_{uv}^{2}\right] = \sum_{\ell} \frac{a\mu(\ell)+b\nu(\ell)}{a+b}\left(\frac{a\mu(\ell)-b\nu(\ell)}{a\mu(\ell)+b\nu(\ell)}\right)^2=\frac{\tau}{\lambda},
\end{eqnarray*}
so that we have for $\tau<1$,
\begin{eqnarray*}
\EE_0\left[ \sum_{i\in \partial\mathcal{T}_d}\Theta_i^2\right] =
|\partial\mathcal{T}_d|\left(\frac{\tau}{\lambda}\right)^{d}\rightarrow 0,
\end{eqnarray*}
as $d$ tends to infinity. Hence by Fatou's lemma, we have
\begin{eqnarray*}
\lim\inf_d\sum_{i\in \partial\mathcal{T}_d}\Theta_i^2 =0 \mbox{ a.s.},
\end{eqnarray*}
and our second claim holds.

Our first claim will hold, once we prove that for $\tau>1$, we have
$\mathcal{R}_{\hbox{eff}}(r\leftrightarrow \infty) =\sup_{d\geq 1}\mathcal{R}_{\hbox{eff}}(r\leftrightarrow
  \partial\mathcal{T}_d)<\infty$.
This fact follows indeed from a computation done in \cite{lp92}. Define the resistance $R'_{ij}=\prod_{(uv)\in\hbox{path }(r\sim
  j)}\theta_{uv}^{-2}$. Note that in our framework the labels of the
edges are i.i.d. with distribution
$\frac{a\mu(\ell)+b\nu(\ell)}{a+b}$. In particular the random variables
$\theta_{uv}$ are also i.i.d. and since $\theta_{uv}\leq 1$, we have
$\min_{0\leq x\leq
  1}\EE_0\left[\theta_{uv}^{2x}\right]=\EE_0\left[\theta_{uv}^{2}\right]$
so that by Theorem 1(i) of \cite{lp92}, for $\tau>1$, we have $\mathcal{R'}_{\hbox{eff}}(r\leftrightarrow
\infty)<\infty$ a.s. Since $R_{uv}\leq R_{uv}'$, we have by Rayleigh's
monotonicity law (see \cite{lyons-book}),
$\mathcal{R}_{\hbox{eff}}(r\leftrightarrow \infty) \leq \mathcal{R'}_{\hbox{eff}}(r\leftrightarrow
\infty)<\infty$ a.s.

\end{proof} 
\section{Numerical results}%
\label{sec:numerical}
\begin{figure}[htb]
\vspace{-0.5cm}\hspace{-1cm}        \begin{minipage}[t]{8cm}
\includegraphics[width=8cm]%
        {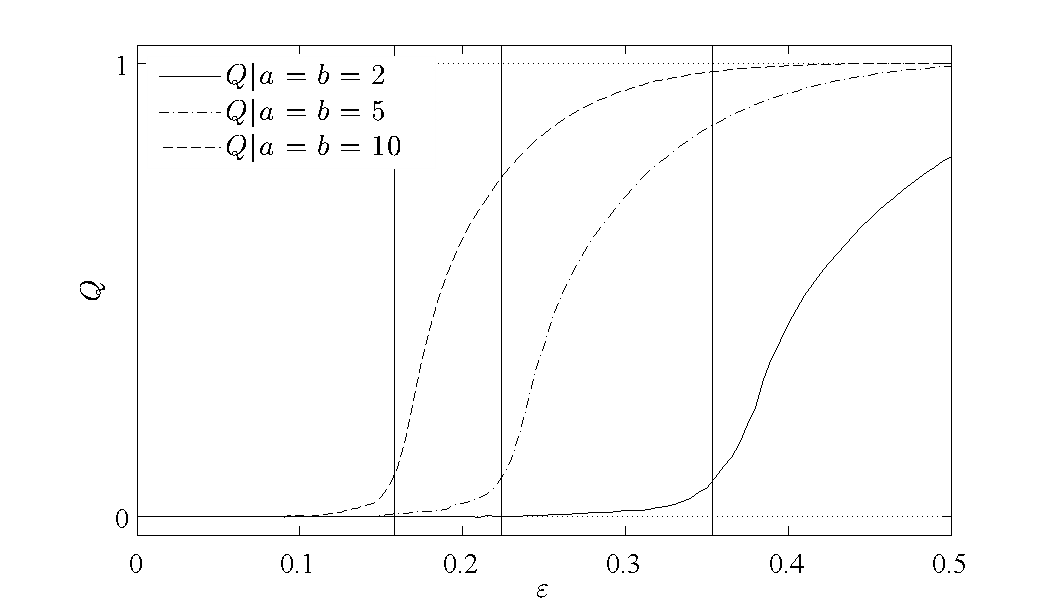}
\end{minipage}
\hfill
\begin{minipage}[t]{8cm}
        \includegraphics[width=8cm]%
        {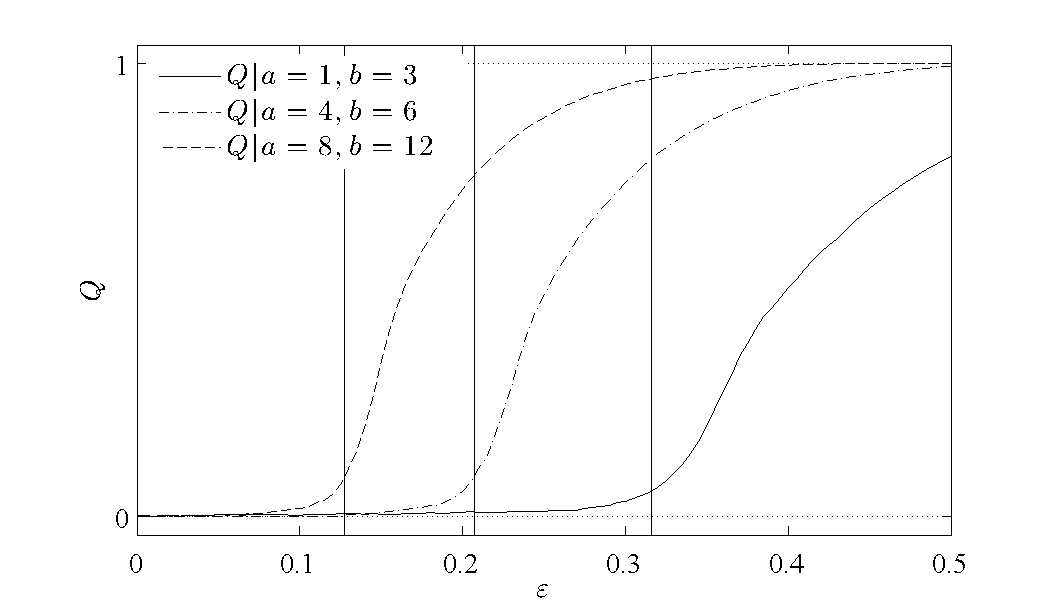}
\end{minipage}
    \caption{Overlap $Q$ as a function of the parameter $\varepsilon$
      (left: $a=b$; right $a<b$)}
    \label{fig:reconstruction-var-eps}
\end{figure}

We now investigate numerically the validity of our proposed conjecture.
We consider first a labelled stochastic block model with two symmetric blocks where the connectivity parameters are identical, i.e. $a=b$, so that community detection can only succeed based on the labels.
We assume for simplicity only two labels $+$ and $-$ and define the distributions $\mu(+) = p$ for edges among nodes of the same type and $\nu(+) = q$ for edges between nodes of different type for two parameters $p,q\in[0,1]$.
In this case, Condition (\ref{cond1}) does not hold, yet reconstruction may still be feasible depending on the values of $p$ and $q$.
In order to validate our conjecture that if the value $\tau$ given in \eqref{eq:tau} is greater than 1, reconstruction may be feasible, we parametrize $p=\frac{1}{2} + \varepsilon$ and $q=\frac{1}{2} - \varepsilon$, which leads to the simplified condition for reconstruction:
\begin{equation}
    \varepsilon > \frac{1}{2\sqrt{a}}.
   \label{eq:thresh-twolabel}
\end{equation}

We characterize the success of the reconstruction using the
\emph{overlap} metric introduced by Decelle et al. in equation (5) of
\cite{Decelle:2011ve}, which we repeat below:
\begin{equation}
    Q(\{\sigma_i\},\{\hat{\sigma}_i\}) = \max_{\pi} \frac{\frac{1}{n} \sum_i \delta_{\sigma_i,\pi(\hat{\sigma}_i)} - \max_{t} n_t}
    {1 - \max_a n_a},
\end{equation}
where $\sigma_i$ denotes the original assignment of types to nodes $i=1\dotsc n$, $\hat{\sigma}_i$,  denotes the estimated assignment, $t$ denotes communities, and $n_t$ is the size of community $t$. In our setup, $t=0$ or 1 and $n_t=n/2$.
Since types may be assigned in different order in the estimate, we
vary over all permutations $\pi(\hat{\sigma}_i)$ of $\hat{\sigma}_i$
and take the one with maximum overlap.
This overlap metric ranges from 0 to 1, equating zero when classification is no better than assigning all nodes to a fixed class (or equivalently, assigning nodes to a randomly chosen type).
We generate a labelled stochastic block model graph with the parameters given above and $n= 5000$ nodes. Then, we use the standard sum--product belief propagation algorithm to infer the types of the nodes based on the labels.
We vary both the density, i.e. $a=b$, and $\varepsilon$. All plotted values are averages over several different seeds.

In Fig.~\ref{fig:reconstruction-var-eps} (left), we plot the overlap metric $Q$ against  $\varepsilon$ on the $x$-axis for $a=b$ given by 2, 5, 10.
For each curve, we indicate the threshold \eqref{eq:thresh-twolabel} as a vertical line in the same style as the corresponding curve.
We observe that to the left of the threshold, $Q$ remains around zero and the variation may be attributed to the initial conditions and small-scale effects.
To the right of the threshold, however, $Q$ increases steadily.

For comparison, in Fig.~\ref{fig:reconstruction-var-eps} (right), we provide the same metric but with $a<b$ given by $(a,b) = (1,3),(4,6),(8,12)$. Accordingly, belief propagation can now exploit both edges as well as their labels and the corresponding curves are shifted towards the left, along with the threshold of $\varepsilon$ where $\tau=1$, again indicated by a vertical line for each curve.

It is interesting, that even for reasonably small scales, belief propagation consistently fails below the threshold, with overlap close to zero, yet achieves positive overlap above the threshold.

%
 
\section{Concluding remarks}%
\label{sec:conclusion}
We have initiated an analysis of community detection in the context of labelled interactions. We have formulated a conjecture on when detectability is feasible, in the form of Condition (\ref{eq:tau}). While restricted to the two symmetric communities case, this condition is already useful in determining how the availability of labels affects detectability. 
A natural extension will consider richer scenarios with more communities, where our techniques can potentially characterize the corresponding transition thresholds. 
On the theoretical front, we have established that two phase transitions, namely sensitivity of belief propagation, and tree reconstructability, coincide in the case of labelled trees. The main outstanding question there is to validate our conjecture that these thresholds characterize the onset of community detectability.



\bibliographystyle{plain}


\bibliography{detect}

\end{document}